\documentclass[pra,twocolumn,superscriptaddress,
showpacs,preprintnumbers,amsmath,amssymb]{revtex4}
\usepackage{graphicx}% Include figure files
\usepackage{subfigure}
\usepackage{amsthm}
\usepackage{tensor}
\usepackage{color}
\usepackage[all]{xy}
\usepackage{tikz}
\usepackage{dsfont}
\usepackage{times,txfonts}
\usetikzlibrary{positioning}
\usepackage{braket}

\newtheorem{Thm}{Theorem}
\newtheorem{Lem}[Thm]{Lemma}
\newtheorem{Prop}[Thm]{Proposition}
\newtheorem{Cor}[Thm]{Corollary}
\theoremstyle{definition}

\begin{document}

\title{One-Shot Quantum State Exchange}

\author{Yonghae Lee} \email{yonghaelee@khu.ac.kr}
\affiliation{
Department of Mathematics and Research Institute for Basic Sciences,
Kyung Hee University, Seoul 02447, Korea}

\author{Hayata Yamasaki} \email{yamasaki@qi.t.u-tokyo.ac.jp}
\affiliation{
Photon Science Center, Graduate School of Engineering, The University of Tokyo,
Bunkyo-ku, Tokyo 113-8656, Japan}

\author{Gerardo Adesso} \email{gerardo.adesso@nottingham.ac.uk}
\affiliation{
School of Mathematical Sciences
and Centre for the Mathematics and Theoretical Physics of Quantum Non-Equilibrium Systems,
University of Nottingham, University Park, Nottingham NG7 2RD, United Kingdom}

\author{Soojoon Lee}\email{level@khu.ac.kr}
\affiliation{
Department of Mathematics and Research Institute for Basic Sciences,
Kyung Hee University, Seoul 02447, Korea}
\affiliation{
School of Mathematical Sciences
and Centre for the Mathematics and Theoretical Physics of Quantum Non-Equilibrium Systems,
University of Nottingham, University Park, Nottingham NG7 2RD, United Kingdom}

\date{\today}

%%%%%%%%%%%%%%%%%%%%%%%%%%%%%%%%%%%%%%%%%%%%%%%%%%%%%%%%%%%%%%%%%%%%%%%%%%%%%%%%%%%%%%%%%%%%%%
%%%
%%%     Abstract
%%%
%%%%%%%%%%%%%%%%%%%%%%%%%%%%%%%%%%%%%%%%%%%%%%%%%%%%%%%%%%%%%%%%%%%%%%%%%%%%%%%%%%%%%%%%%%%%%%
\begin{abstract}
The quantum state exchange is a quantum communication task
in which two users exchange their respective quantum information in the asymptotic setting.
In this work,
we consider a one-shot version of the quantum state exchange task,
in which the users hold a single copy of the initial state,
and they exchange their parts of the initial state
by means of entanglement-assisted local operations and classical communication.
We first derive lower bounds on the least amount of entanglement required
for carrying out this task,
and provide conditions on the initial state such that the protocol succeeds with zero entanglement cost.
Based on these results,
we reveal two counter-intuitive phenomena in this task, which make it different from a conventional SWAP operation.
One tells how the users deal with their symmetric information
in order to reduce the entanglement cost.
The other shows that it is possible for the users to gain
extra shared entanglement after this task.
\end{abstract}

\maketitle

%%%%%%%%%%%%%%%%%%%%%%%%%%%%%%%%%%%%%%%%%%%%%%%%%%%%%%%%%%%%%%%%%%%%%%%%%%%%%%%%%%%%%%%%%%%%%%
%%%
%%%     Introduction
%%%
%%%%%%%%%%%%%%%%%%%%%%%%%%%%%%%%%%%%%%%%%%%%%%%%%%%%%%%%%%%%%%%%%%%%%%%%%%%%%%%%%%%%%%%%%%%%%%
{\it Introduction.}---
In quantum information theory,
the quantum state exchange~\cite{OW08,LTYAL19} is a quantum communication task,
in which two users, Alice and Bob, exchange their quantum information
by means of local operations and classical communication (LOCC) assisted by shared entanglement.
A main research aim in the study of the quantum state exchange is to evaluate
the least amount of entanglement needed for the task,
as in other quantum communication tasks,
such as quantum state merging~\cite{HOW05,HOW06} and quantum state redistribution~\cite{DY08,YD09}.

Most quantum communication tasks~\cite{S95,HOW05,HOW06,DY08,YD09,ADHW09}
including the quantum state exchange usually assume the {\em asymptotic scenario},
in which users can have an unbounded number of
independent and identically distributed copies of an initial state,
and they carry out their task with the copies.
On the other hand,
it is not easy in a realistic situation
to prepare a sufficiently large number of state copies,
and the amount of non-local resources available for the users is limited.
To reflect these practical difficulties, quantum information research has focused more recently on
the {\em one-shot scenario}~\cite{BCR11,RR12,WR12,DH13,BCT16,ZLYCM18,AJW18,RFWA18,YM19}.

In this work,
we introduce and study the {\em one-shot quantum state exchange} (OSQSE) task. This is not only a useful quantum communication task, but can also
have a potential application in quantum computation.
Let us consider a specific situation as follows.
Alice and Bob want to carry out the SWAP gate~\cite{NC10},
which plays an important role in universal quantum computation~\cite{JM08}.
The problem is that they cannot directly apply the SWAP gate,
because they are far apart.
If Alice and Bob are sharing prior entanglement,
then the OSQSE can be a method to non-locally perform the SWAP gate, as both
operationally provide the same result. However, the OSQSE has unique properties which we reveal in this work.

We formally define the OSQSE and its optimal entanglement cost,
and derive computable lower bounds on the latter,
which in turn yield bounds for the asymptotic quantum state exchange~\cite{OW08,LTYAL19}.
In addition,
we provide two useful conditions to decide whether a given initial state
enables OSQSE with zero entanglement cost.
%In the one-shot setting,
%it is not easy to evaluate optimal resource costs in general,
%but we can easily check whether the optimal entanglement cost is zero or not, from the conditions.
We then show that there exist counter-intuitive phenomena in the OSQSE task that set it apart from the conventional SWAP operation.

%%%%%%%%%%%%%%%%%%%%%%%%%%%%%%%%%%%%%%%%%%%%%%%%%%%%%%%%%%%%%%%%%%%%%%%%%%%%%%%%%%%%%%%%%%%%%%
%%%
%%%     One-shot quantum state exchange
%%%
%%%%%%%%%%%%%%%%%%%%%%%%%%%%%%%%%%%%%%%%%%%%%%%%%%%%%%%%%%%%%%%%%%%%%%%%%%%%%%%%%%%%%%%%%%%%%%
{\it One-shot quantum state exchange.}---
Consider two users, Alice and Bob, holding parts $A$ and $B$
of the initial state $\ket{\psi}\equiv\ket{\psi}_{A_1B_1A_2B_2R}$ with systems $A=A_1A_2$ and $B=B_1B_2$,
respectively.
Alice's and Bob's goal is either to exchange their parts $A_1$ and $B_1$
or to exchange their whole parts $A$ and $B$.

Specifically,
let $\psi_{f_1}$ and $\psi_{f_{12}}$ be the final states of the task, %given by
\begin{eqnarray}
\psi_{f_1}&=&
\big(
\mathds{1}_{A_1\to A'_1}
\otimes
\mathds{1}_{B_1\to B'_1}
\otimes
\mathds{1}_{A_2B_2R}
\big)(\psi), \nonumber \\
\psi_{f_{12}}&=&
\big(
\mathds{1}_{A\to A'}
\otimes
\mathds{1}_{B\to B'}
\otimes
\mathds{1}_{R}
\big)(\psi), \nonumber
\end{eqnarray}
where $\psi=\ket{\psi}\bra{\psi}$,
and the dimension of system $X'$ is identical to that of system $X$.
Note that $B'_1$, $B'$ and $A'_1$, $A'$ are Alice's and Bob's systems,
respectively.
Then three joint operations
\begin{eqnarray}\label{eq:OSQSE}
&&\mathcal{E}_{\psi,K,L}^1:
A_1E_{\mathrm{A}}^{\mathrm{in}}\otimes
B_1E_{\mathrm{B}}^{\mathrm{in}}
\longrightarrow B'_1E_{\mathrm{A}}^{\mathrm{out}}\otimes
A'_1E_{\mathrm{B}}^{\mathrm{out}}, \nonumber \\
&&\mathcal{E}_{\psi,K,L}^{1|2}:
AE_{\mathrm{A}}^{\mathrm{in}}\otimes
BE_{\mathrm{B}}^{\mathrm{in}}
\longrightarrow B'_1A_2E_{\mathrm{A}}^{\mathrm{out}}\otimes
A'_1B_2E_{\mathrm{B}}^{\mathrm{out}},  \\
&&\mathcal{E}_{\psi,K,L}^{12}:
AE_{\mathrm{A}}^{\mathrm{in}}\otimes
BE_{\mathrm{B}}^{\mathrm{in}}
\longrightarrow B'E_{\mathrm{A}}^{\mathrm{out}}\otimes
A'E_{\mathrm{B}}^{\mathrm{out}}, \nonumber
\end{eqnarray}
are called the OSQSE protocols
of $\ket{\psi}$,
if they are performed by LOCC between Alice and Bob,
and satisfy
\begin{eqnarray} \label{eq:PDs} % Protocol Definitions
\psi_{f_1} \otimes {\Phi}
&=&
\left(\mathcal{E}_{\psi,K,L}^1\otimes\mathds{1}_{A_2B_2R}\right)
\left( \psi \otimes {\Psi} \right)
\nonumber \\
&=&
\left(\mathcal{E}_{\psi,K,L}^{1|2}\otimes\mathds{1}_{R}\right)
\left( \psi \otimes {\Psi} \right),
 \\
\psi_{f_{12}} \otimes {\Phi}
&=&
\left(\mathcal{E}_{\psi,K,L}^{12}\otimes\mathds{1}_{R}\right)
\left( \psi \otimes {\Psi} \right), \nonumber
\end{eqnarray}
where
${\Psi}$ and ${\Phi}$ are pure maximally entangled states
with Schmidt rank $K$ and $L$
on systems $E_{\mathrm{A}}^{\mathrm{in}}E_{\mathrm{B}}^{\mathrm{in}}$
and $E_{\mathrm{A}}^{\mathrm{out}}E_{\mathrm{B}}^{\mathrm{out}}$,
respectively.
It is possible to generalize the above definitions by adding errors
for approximation to Eq.~(\ref{eq:PDs}),
but it suffices to only consider error-free protocols
to obtain our main results.

At this point,
it is instructive to inform differences among the three protocols in Eq.~(\ref{eq:OSQSE}) as follows:
The first two protocols $\mathcal{E}_{\psi,K,L}^1$ and $\mathcal{E}_{\psi,K,L}^{1|2}$ indicate
that only the parts $A_1$ and $B_1$ are exchanged,
while the whole parts $A_1A_2$ and $B_1B_2$ are exchanged
in the third protocol $\mathcal{E}_{\psi,K,L}^{12}$.
In addition,
the parts $A_2$ and $B_2$ can be used for exchanging $A_1$ and $B_1$ in the protocol $\mathcal{E}_{\psi,K,L}^{1|2}$,
while $A_2$ and $B_2$ are untouched in the protocol $\mathcal{E}_{\psi,K,L}^1$.

Depending on the types of OSQSE protocols,
we define three \emph{optimal entanglement costs} % as follows:
\begin{eqnarray}\label{eq:costs}
\mathbf{e}_{A_1\leftrightarrow B_1}\left(\psi\right)
&=& \inf_{\mathcal{E}_{\psi,K,L}^1} \left( \log K-\log L \right), \nonumber \\
\mathbf{e}_{A_1\leftrightarrow B_1}^{A_2B_2}\left(\psi\right)
&=& \inf_{\mathcal{E}_{\psi,K,L}^{1|2}} \left( \log K-\log L \right),  \\
\mathbf{e}_{A\leftrightarrow B}\left(\psi\right)
&=& \inf_{\mathcal{E}_{\psi,K,L}^{12}} \left( \log K-\log L \right), \nonumber
\end{eqnarray}
where the quantity $\log K-\log L$ is called the \emph{entanglement cost} of the OSQSE protocol,
and the infimums are taken over all joint protocols
$\mathcal{E}_{\psi,K,L}^1$, $\mathcal{E}_{\psi,K,L}^{1|2}$, and $\mathcal{E}_{\psi,K,L}^{12}$, respectively.

By the definitions of the optimal entanglement costs,
we obtain the following proposition.

\begin{Prop} \label{prop:OEC1} For any input state $\psi$,
$\mathbf{e}_{A_1\leftrightarrow B_1}\left(\psi\right)
\ge\mathbf{e}_{A_1\leftrightarrow B_1}^{A_2B_2}\left(\psi\right)$.
\end{Prop}

%%%%%%%%%%%%%%%%%%%%%%%%%%%%%%%%%%%%%%%%%%%%%%%%%%%%%%%%%%%%%%%%%%%%%%%%%%%%%%%%%%%%%%%%%%%%%%
%%%
%%%     Converse bound (Lower bound)
%%%
%%%%%%%%%%%%%%%%%%%%%%%%%%%%%%%%%%%%%%%%%%%%%%%%%%%%%%%%%%%%%%%%%%%%%%%%%%%%%%%%%%%%%%%%%%%%%%
{\it Converse bounds.}---
A real number $r$ is called a \emph{converse bound} of the optimal entanglement cost if
it is upper bounded by the entanglement cost of any OSQSE protocol.
We derive converse bounds of the optimal entanglement costs as follows.%for the initial state $\ket{\psi}$.

As in the asymptotic scenario~\cite{OW08,LTYAL19},
we consider a one-shot version of the $R$-assisted quantum state exchange task,
in which the reference system $R$ is divided
into two systems $R_{\mathrm{A}}$ and $R_{\mathrm{B}}$,
and then Alice and Bob receive the divided parts $R_{\mathrm{A}}$ and $R_{\mathrm{B}}$,
respectively,
so that the initial state $\ket{\tilde{\psi}}_{A_1B_1A_2B_2R_{\mathrm{A}}R_{\mathrm{B}}}$ is divided
into Alice's parts $AR_{\mathrm{A}}$ and Bob's parts $BR_{\mathrm{B}}$.
This can be realized by using a quantum channel $\mathcal{N}:R\longrightarrow R_{\mathrm{A}}$
and its complementary channel $\mathcal{N}^{c}:R\longrightarrow R_{\mathrm{B}}$~\cite{W13}.
Let $\mathcal{E}_{\tilde{\psi},K,L}$ be an $R$-assisted OSQSE protocol,
\begin{equation*}
\mathcal{E}_{\tilde{\psi},K,L}:
AR_{\mathrm{A}}E_{\mathrm{A}}^{\mathrm{in}}\otimes
BR_{\mathrm{B}}E_{\mathrm{B}}^{\mathrm{in}} \nonumber \longrightarrow B'_1A_2R_{\mathrm{A}}E_{\mathrm{A}}^{\mathrm{out}}\otimes
A'_1B_2R_{\mathrm{B}}E_{\mathrm{B}}^{\mathrm{out}},
\end{equation*}
with the entanglement cost $\log K-\log L$ such that
$\tilde{\psi}_f \otimes {\Phi} =\mathcal{E}_{\tilde{\psi},K,L}
( \tilde{\psi} \otimes {\Psi} )$,
where $\tilde{\psi}_f=
(
\mathds{1}_{A_1\to A'_1}
\otimes
\mathds{1}_{B_1\to B'_1}
\otimes
\mathds{1}_{A_2B_2R_{\mathrm{A}}R_{\mathrm{B}}}
)(\tilde{\psi})$.
%which exchanges $A_1$ and $B_1$.
%by using Alice' and Bob's QSI $A_2R_{\mathrm{A}}$ and $B_2R_{\mathrm{B}}$.
Note that $\mathcal{E}_{\tilde{\psi},K,L}$ is an LOCC protocol by Alice and Bob.
Let $\sigma^M$ be the maximally mixed state with rank $M$.
From the majorization condition for LOCC convertibility~\cite{N99,MGH14},
the state $\rho_{B'_1A_2R_{\mathrm{A}}}\otimes\sigma^L_{E_{\mathrm{B}}^{\mathrm{out}}}$
majorizes
the state $\rho_{AR_{\mathrm{A}}}\otimes\sigma^K_{E_{\mathrm{B}}^{\mathrm{in}}}$.
Let $F$ be an additive and Schur concave function~\cite{DH02} %of a quantum state $\rho$
such that $F(\sigma^M)=\log M$ for any $M$.
%for any maximally mixed state $\sigma$ with rank $T$.
From the Schur concavity of the function $F$,
the inequality
%\begin{equation} \label{eq:renyi_ineq}
$F(
\rho_{B'_1A_2R_{\mathrm{A}}}\otimes\sigma^L_{E_{\mathrm{B}}^{\mathrm{out}}}
)
\le
F(
\rho_{AR_{\mathrm{A}}}\otimes\sigma^K_{E_{\mathrm{B}}^{\mathrm{in}}}
)$
%\nonumber
%\end{equation}
holds.
Since $F$ is additive and $\rho_{B'_1A_2R_{\mathrm{A}}}=\rho_{B_1A_2R_{\mathrm{A}}}$,
it follows that
$\log K - \log L \ge F(\rho_{B_1A_2R_{\mathrm{A}}})-F(\rho_{AR_{\mathrm{A}}})$.
%\begin{equation}
%\log K - \log L
%\ge
%F(\rho_{B'_1A_2R_{\mathrm{A}}})-F(\rho_{AR_{\mathrm{A}}}). \nonumber
%\end{equation}
Since any protocol $\mathcal{E}_{\psi,K,L}^{1|2}$
is also an $R$-assisted OSQSE protocol for the initial state $\ket{\psi}$,
we obtain the following theorem.

\begin{Thm} \label{thm:GLB}
For any input state $\psi$, the optimal entanglement cost $\mathbf{e}_{A_1\leftrightarrow B_1}^{A_2B_2}\left(\psi\right)$
is lower bounded by
\begin{equation}
l_{1|2}(\psi)=
\sup_{F,\mathcal{N}}
\left[
F(\mathcal{N}(\psi)_{B_1A_2R_{\mathrm{A}}})-F(\mathcal{N}(\psi)_{AR_{\mathrm{A}}})
\right], \nonumber
\end{equation}
where $F$ is an additive and Schur concave function
such that $F(\sigma^M)=\log M$ for any $M$
and $\mathcal{N}(\rho)$ is a quantum channel from $R$ to $R_{\mathrm{A}}$.
\end{Thm}

In Theorem~\ref{thm:GLB},
if $R$ is directly sent to either Alice or Bob without splitting,
and we restrict the function $F$ to the quantum R\'enyi entropy $S_{\alpha}(\varrho)$
of order $\alpha$~\cite{DH02}
for a quantum state $\varrho$,
then we obtain the following computable converse bounds.

\begin{Cor} \label{cor:CLB}
%\begin{equation}
For any input state $\psi$, $\mathbf{e}_{A_1\leftrightarrow B_1}^{A_2B_2}\left(\psi\right)
\ge
l_{1|2}^c(\psi)
=\max_{\alpha\in[0,\infty]} f_{\psi}(\alpha)$,% \nonumber \\
%\end{equation}
where $f_{\psi}(\alpha)$ is a function
of $\ket{\psi}$ and $\alpha$ defined by
%\begin{equation}
$f_{\psi}(\alpha)=
\max\{ S_{\alpha}(\rho_{A_1B_2})-S_{\alpha}(\rho_{B}),
S_{\alpha}(\rho_{B_1A_2})-S_{\alpha}(\rho_{A}) \}$.
%\end{equation}
\end{Cor}

We refer the reader to Appendix~\ref{app:pf_CLB} for the proof of Corollary~\ref{cor:CLB}.
Remark that
the converse bound $l_{1|2}^c$
can be easily computed
by means of analytical or numerical methods,
since the function $f_{\psi}(\alpha)$ is one-variable and differentiable on $(0,\infty)$.
For the different types of the OSQSE protocols,
we can also obtain a similar computable converse bound as follows:
\begin{equation} \label{eq:CLB}
\mathbf{e}_{X\leftrightarrow Y}\left(\psi\right)
\ge
\max_{\alpha\in[0,\infty]} \left| S_{\alpha}(\rho_{X})-S_{\alpha}(\rho_{Y}) \right|,
\end{equation}
where the pair  $(X,Y)$ can refer either to $(A_1,B_1)$ or to $(A,B)$.

We also remark that
in Theorem~\ref{thm:GLB},
if $F$ is chosen as the von Neumann entropy~\cite{W13},
then the converse bound $l_{1|2}$ recovers a theoretical converse bound in Refs.~\cite{OW08,LTYAL19}.
In addition,
a computable converse bound therein is just $f_{\psi}(1)$ in Corollary~\ref{cor:CLB}.
By virtue of the additivity of $F$,
it is clear that
$l_{1|2}$ and $l_{1|2}^c$ are also
converse bounds of the optimal entanglement cost
for the asymptotic quantum state exchange task.
Hence,
our converse bounds improve the existing bounds in Refs.~\cite{OW08,LTYAL19}.
For example,
if the initial state $\ket{\psi_1}\equiv\ket{\psi_1}_{A_1B_1A_2B_2R}$ is
\begin{equation} \label{eq:graph}
\ket{\psi_1}
=\mbox{$
\frac{1}{5}\ket{00000}+\sqrt{\frac{3}{50}}\ket{00010}
+\frac{3}{5}\ket{01001}+\sqrt{\frac{27}{50}}\ket{11100}$},
\end{equation}
then we can find a value $\alpha_0\in[0,\infty]$
such that $l_{1|2}^c(\psi_1)=f_{\psi_1}(\alpha_0)>f_{\psi_1}(1)$
as depicted in Fig.~\ref{fig:max_graph}.
This example shows that our  bound $l_{1|2}^c(\psi)$ is tighter
than the existing bound $f_{\psi}(1)$.

\begin{figure}
\centering
\includegraphics[width=.8\linewidth,trim=0cm 0cm 0cm 0cm]{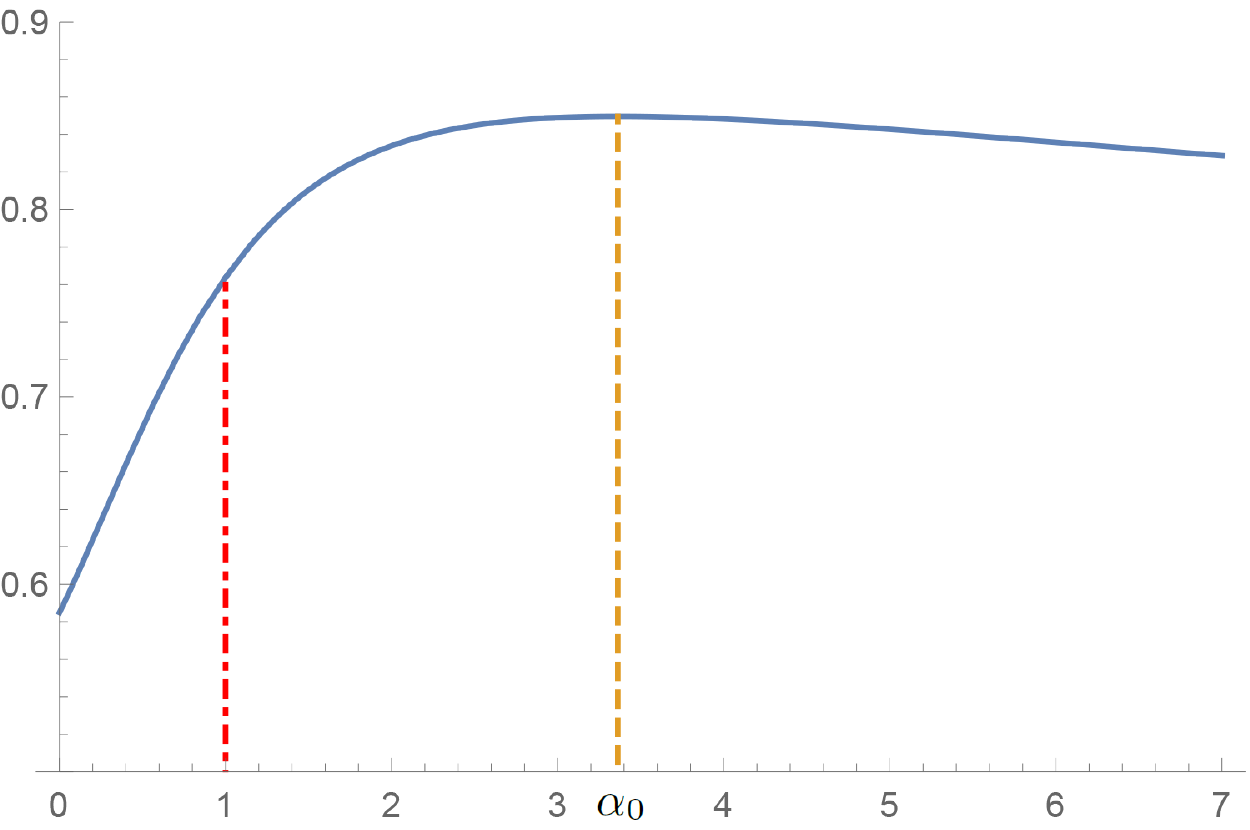}
\caption{
The graph of the function $f_{\psi_1}(\alpha)$
for a specific initial state $\ket{\psi_1}$ in Eq.~(\ref{eq:graph}).
The maximum of $f_{\psi_1}(\alpha)$ is attained
at the point $\alpha_0$ $(\approx 3.362)$.
In the graph,
$\alpha_0$ is represented as the yellow dashed line,
and 1 is represented as the red dashed dotted line.
%$f_{\psi_1}(\alpha_0)\approx 0.8496$ and
%$f_{\psi_1}(1)\approx 0.7637$.
}
\label{fig:max_graph}
\end{figure}

%%%%%%%%%%%%%%%%%%%%%%%%%%%%%%%%%%%%%%%%%%%%%%%%%%%%%%%%%%%%%%%%%%%%%%%%%%%%%%%%%%%%%%%%%%%%%%
%%%
%%%     Conditions for zero entanglement cost
%%%
%%%%%%%%%%%%%%%%%%%%%%%%%%%%%%%%%%%%%%%%%%%%%%%%%%%%%%%%%%%%%%%%%%%%%%%%%%%%%%%%%%%%%%%%%%%%%%
{\it Conditions for zero entanglement cost.}---
We now present conditions for OSQSE with zero entanglement cost.

By the converse bound in Eq.~(\ref{eq:CLB}),
it is obvious that if there exist Alice's and Bob's local isometries performing the OSQSE task,
then the optimal entanglement cost is zero.
We first characterize this type of strategy.
Let $(X,Y)$ be a pair of two systems,
which can be either $(A_1,B_1)$ or $(A,B)$,
and consider a spectral decomposition
of the reduced state $\rho_{XY}$ for $\ket{\psi}$,
$\rho_{XY}=\sum_{i=1}^{N}\lambda_i\ket{\xi_i}\bra{\xi_i}_{XY}$,
%\begin{equation}
%\rho_{XY}=\sum_{i=1}^{N}\lambda_i\ket{\xi_i}\bra{\xi_i}_{XY}, \nonumber
%\end{equation}
where $\lambda_i>0$ with $\sum_{i=1}^{N}\lambda_i=1$.
For each $i$,
we define the matrix $\Omega_{XY}^i(\psi)$ by
\begin{equation}
\Omega_{XY}^i(\psi)
=\sum_{j,k}\left(\bra{j}_{X}\otimes\bra{k}_{Y}\right)\ket{\xi_i}_{XY} \ket{j}\bra{k}, \nonumber
\end{equation}
where $\{\ket{j}\}$ and $\{\ket{k}\}$ indicate the computational bases on Alice's and Bob's systems,
respectively.
Then we obtain the following sufficient condition.

\begin{Thm} \label{thm:SC0}
Let $(X,Y)$ be either $(A_1,B_1)$ or $(A,B)$.
If there exist isometries $U$ and $V$ such that, for each $i$,
${\left(\Omega_{XY}^i(\psi)\right)}^t=U \Omega_{XY}^i(\psi)V$,
then $\mathbf{e}_{X\leftrightarrow Y}\left(\psi\right)=0$.
%, where $\cdot^t$ denotes transposition.
\end{Thm}

Here, the isometries $U$ and $V$ indicate Alice's and Bob's local operations
exchanging the parts $X$ and $Y$ without shared entanglement.
The proof of Theorem~\ref{thm:SC0} is in Appendix~\ref{app:pf_thm:SC0}.

From the converse bound in Eq.~(\ref{eq:CLB}),
observe that if the spectrum of Alice's state is different from that of Bob's state,
then the optimal entanglement cost cannot be zero.
Based on this observation,
we obtain the following theorem,
whose proof can be found in Appendix~\ref{app:pf_thm:SCP}.

\begin{Thm} \label{thm:SCP}
Let $(X,Y)$ be either $(A_1,B_1)$ or $(A,B)$.
If $\mathbf{e}_{X\leftrightarrow Y}\left(\psi\right)=0$,
then there exists an isometry $U_{X\rightarrow Y}$ such that
$\rho_{Y}= U_{X\rightarrow Y}\rho_{X}(U_{X\rightarrow Y})^{\dagger}$.
%\begin{equation}
%\rho_{Y}= U_{X\rightarrow Y}\rho_{X}(U_{X\rightarrow Y})^{\dagger}. \nonumber
%\end{equation}
\end{Thm}

We remark that
the converse of Theorem~\ref{thm:SCP} is not true in general.
Let us consider the following simple initial state
\begin{equation}
\ket{\psi_2}_{A_1B_1A_2B_2}
=\mbox{$\frac{1}{2}\left(\ket{0000}+\ket{0101}+\ket{1010}+\ket{1111}\right)$}, \nonumber
\end{equation}
then
$\mathbf{e}_{A_1\leftrightarrow B_1}\left(\psi_2\right)
\ge\mathbf{e}_{A_1\leftrightarrow B_1}^{A_2B_2}\left(\psi_2\right)
\ge 2$,
from Proposition~\ref{prop:OEC1} and Corollary~\ref{cor:CLB}.
However,
the state $\ket{\psi_2}$ satisfies the necessary condition in Theorem~\ref{thm:SCP},
since its reduced states $\rho_{A_1}$ and $\rho_{B_1}$ are identical.

%%%%%%%%%%%%%%%%%%%%%%%%%%%%%%%%%%%%%%%%%%%%%%%%%%%%%%%%%%%%%%%%%%%%%%%%%%%%%%%%%%%%%%%%%%%%%%
%%%
%%%     Counter-intuitive phenomena
%%%
%%%%%%%%%%%%%%%%%%%%%%%%%%%%%%%%%%%%%%%%%%%%%%%%%%%%%%%%%%%%%%%%%%%%%%%%%%%%%%%%%%%%%%%%%%%%%%
{\it Counter-intuitive phenomena.}---
We are now in the position to present two phenomena
which show the important differences between the OSQSE task and the SWAP operation.

%%%%%%%%%%%%%%%%%%%%%%%%%%%%%%%%%%%%%%%%%%%%%%%%%%%%%%%%%%%%%%%%%%%%%%%%%%%%%%%%%%%%%%%%%%%%%%
%%%     (1) Symmetric information
%%%%%%%%%%%%%%%%%%%%%%%%%%%%%%%%%%%%%%%%%%%%%%%%%%%%%%%%%%%%%%%%%%%%%%%%%%%%%%%%%%%%%%%%%%%%%%
{\it (1) Symmetric information.}---
For the initial state $\ket{\psi}$,
let us consider a scenario in which Alice and Bob exchange their whole information $A$ and $B$.
Assume that their parts $A_2$ and $B_2$ are symmetric,
while the remaining parts $A_1$ and $B_1$ are not symmetric,
i.e.,
the initial state $\ket{\psi}$ satisfies
%\begin{equation}
$\left(\mathrm{SWAP}_{A_1\leftrightarrow B_1}\right)\left(\psi\right)\neq\psi$ and
%\quad \mathrm{and} \quad
$\left(\mathrm{SWAP}_{A_2\leftrightarrow B_2}\right)\left(\psi\right)=\psi$
%, \nonumber
%\end{equation}
where $\mathrm{SWAP}_{X\leftrightarrow Y}$ is the operation swapping
quantum states in systems $X$ and $Y$.

From a viewpoint of the SWAP operation,
if Alice and Bob want to exchange $A$ and $B$,
then it suffices for them to exchange $A_1$ and $B_1$,
since $A_2$ is identical to $B_2$.
This situation can be more easily understood
by using a cargo exchange as a metaphor for the SWAP operation
as depicted in Fig.~\ref{fig:Cargo}.
In the cargo exchange,
assume that Alice and Bob want to exchange their whole cargoes,
and some of the cargoes are symmetric.
In terms of efficiency,
it is reasonable for them to exchange only $A_1$ and $B_1$
in order to reduce the cargo exchange cost,
because the cargoes $A_2$ and $B_2$ are the same.

\begin{figure}
\centering
\includegraphics[width=.8\linewidth,trim=0cm 0cm 0cm 0cm]{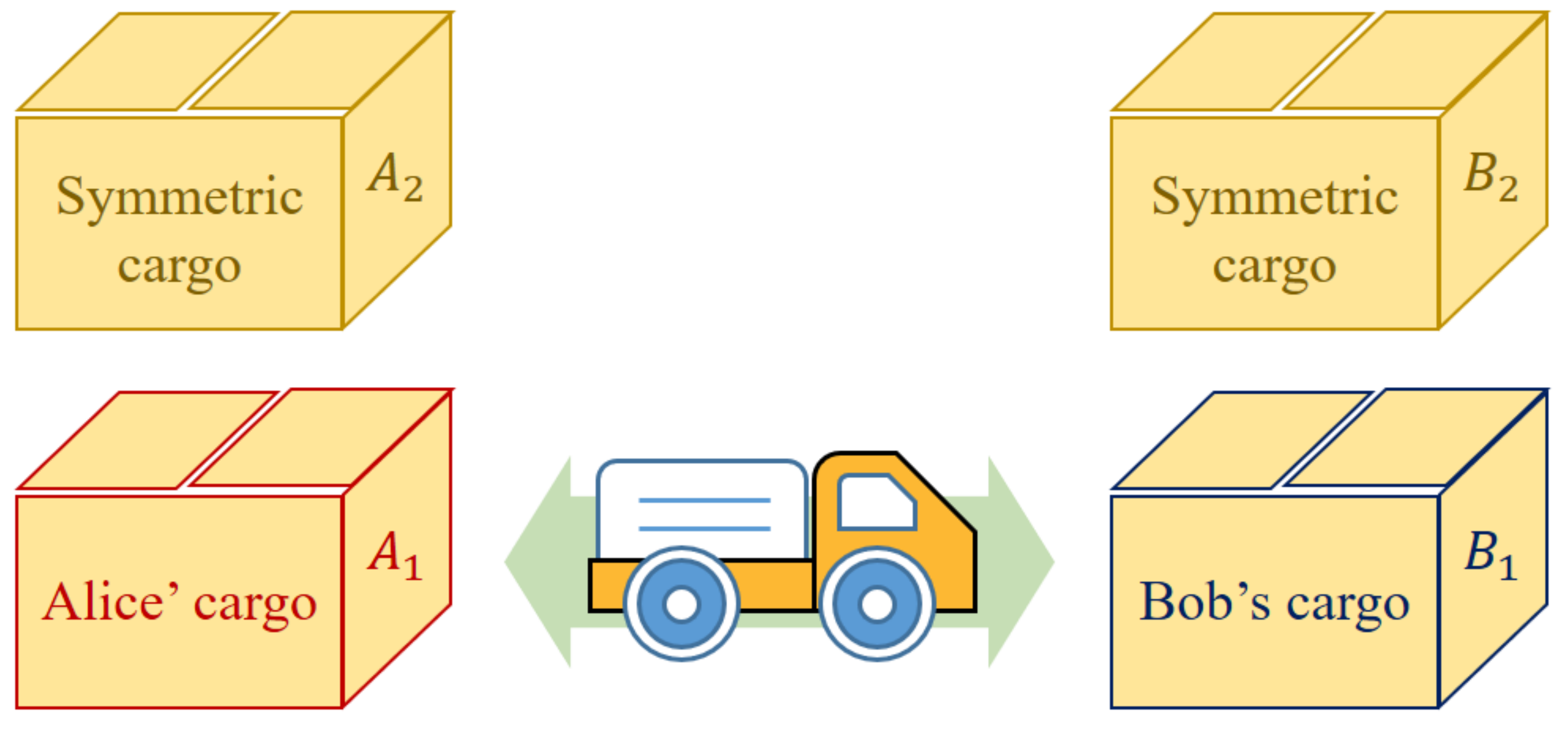}
\caption{
Illustration of the cargo exchange task.
The cargoes $A_1A_2$ and $B_1B_2$ belong to Alice and Bob, respectively.
Assume that the cargoes $A_2$ and $B_2$ are symmetric,
but $A_1$ and $B_1$ are not symmetric.
When Alice and Bob exchange their whole cargoes $A_1A_2$ and $B_1B_2$,
it suffices for them to exchange $A_1$ and $B_1$,
since $A_2$ is identical to $B_2$.
The truck indicates the cost needed for exchanging $A_1$ and $B_1$.
}
\label{fig:Cargo}
\end{figure}

On the other hand,
in the OSQSE,
the proper use of the symmetric parts $A_2$ and $B_2$
can more efficiently reduce the entanglement cost
compared to exchanging only $A_1$ and $B_1$ without using $A_2$ and $B_2$.
To be specific,
there exists an initial state $\ket{\psi}$ such that
the parts $A_2$ and $B_2$ are symmetric
and $\mathbf{e}_{A\leftrightarrow B}\left(\psi\right)=0$
while the rest parts $A_1$ and $B_1$ are not symmetric.
Consider the specific initial state
\begin{equation}
\ket{\phi_1}_{A_1B_1A_2B_2R}
=\mbox{$\frac{1}{\sqrt{2}}
\left(\ket{00000}+\ket{01111} \right)$}, \nonumber
\end{equation}
where $A_2$ and $B_2$ are symmetric but $A_1$ and $B_1$ are not.
Since $\Omega_{AB}^{1}(\phi_1)=\ket{00}\bra{00}$
and $\Omega_{AB}^{2}(\phi_1)=\ket{01}\bra{11}$,
we can show that $\Omega_{AB}^{1}(\phi_1)$ and $\Omega_{AB}^{2}(\phi_1)$ satisfy
the condition in Theorem~\ref{thm:SC0},
by setting
\begin{equation} \label{eq:UV}
U=V
=\ket{00}\bra{00}+\ket{01}\bra{11}+\ket{10}\bra{10}+\ket{11}\bra{01}.
\end{equation}
Thus
we obtain that
$\mathbf{e}_{A\leftrightarrow B}\left(\phi_1\right)=0$,
which means that
$A$ and $B$ can be exchanged by means of LOCC
without consuming any non-local resource.
As mentioned above,
%this is somewhat counter-intuitive,
%since
this phenomenon cannot occur when using the SWAP operation.

The above example also shows that
the use of the symmetric parts $A_2$ and $B_2$ can reduce the entanglement cost
for exchanging $A_1$ and $B_1$.
Since the initial state $\ket{\phi_1}$ does not satisfy the necessary condition in Theorem~\ref{thm:SCP},
we obtain $\mathbf{e}_{A_1\leftrightarrow B_1}\left(\phi_1\right)>0$.
Observe that the isometry $U$ ($V$) in Eq.~(\ref{eq:UV})
represents Alice's (Bob's) local operation
$\mathrm{CNOT}_{\mathrm{A}}$ ($\mathrm{CNOT}_{\mathrm{B}}$)
whose target and controlled systems are $A_1$ ($B_1$) and $A_2$ ($B_2$),
respectively.
This implies that Alice and Bob can exchange $A_1$ and $B_1$
by using local operations.
It follows that $0\ge\mathbf{e}_{A_1\leftrightarrow B_1}^{A_2B_2}\left(\phi_1\right)$.
In fact,
$\mathbf{e}_{A_1\leftrightarrow B_1}^{A_2B_2}\left(\phi_1\right)=0$ from Corollary~\ref{cor:CLB}.
Therefore,
we obtain $\mathbf{e}_{A_1\leftrightarrow B_1}\left(\phi_1\right)
>\mathbf{e}_{A_1\leftrightarrow B_1}^{A_2B_2}\left(\phi_1\right)$.

When $A_2$ and $B_2$ are symmetric,
we can show the following relation between the optimal entanglement costs by definition.

\begin{Prop} \label{prop:OEC2}
$\mathbf{e}_{A\leftrightarrow B}\left(\psi\right)
=\mathbf{e}_{A_1\leftrightarrow B_1}^{A_2B_2}\left(\psi\right)$,
if the parts $A_2$ and $B_2$ of $\ket{\psi}$ are symmetric.
%\begin{equation}
%\mathbf{e}_{A\leftrightarrow B}\left(\psi\right)
%=\mathbf{e}_{A_1\leftrightarrow B_1}^{A_2B_2}\left(\psi\right). \nonumber
%\end{equation}
\end{Prop}

From Proposition~\ref{prop:OEC2},
we can see that,
when Alice and Bob exchange systems $A$ and $B$ of $\ket{\psi}$ with symmetric parts $A_2$ and $B_2$,
they can achieve the optimal entanglement cost
by exchanging only $A_1$ and $B_1$,
making the most of this symmetry.

%%%%%%%%%%%%%%%%%%%%%%%%%%%%%%%%%%%%%%%%%%%%%%%%%%%%%%%%%%%%%%%%%%%%%%%%%%%%%%%%%%%%%%%%%%%%%%
%%%     (2) Negative entanglement cost
%%%%%%%%%%%%%%%%%%%%%%%%%%%%%%%%%%%%%%%%%%%%%%%%%%%%%%%%%%%%%%%%%%%%%%%%%%%%%%%%%%%%%%%%%%%%%%
{\it (2) Negative entanglement cost.}---
As in the asymptotic quantum state exchange task~\cite{OW08,LTYAL19},
there exist initial states to show that
the entanglement cost of the OSQSE task can be negative.
Assume that Alice and Bob exchange the parts $A_1$ and $B_1$
of the initial state
\begin{equation} \label{eq:NOEC}
\ket{\phi_2}_{A_{1}B_{1}A_{2}B_{2}}
=\mbox{$\frac{1}{2}\sum_{i,j=0}^1\ket{i}_{A_{1}}\ket{j}_{B_{1}}\ket{j}_{A_{2}}\ket{i}_{B_{2}}$},
\end{equation}
where $\ket{\phi_2}$ consists of two ebits $\ket{e}_{A_{1}B_{2}}$ and $\ket{e}_{B_{1}A_{2}}$.
To exchange $A_1$ and $B_1$,
both Alice and Bob prepare an ebit,
respectively,
and they locally implement entanglement swapping~\cite{ZZHE93}
by performing two Bell measurements on $A_2$, $B_2$, and the parts of the ebits,
as described in Fig.~\ref{fig:ES}.
Then they can exchange $A_1$ and $B_1$,
and can share two ebits at the same time.
This means that the entanglement cost can be negative.
In fact,
we have $\mathbf{e}_{A_1\leftrightarrow B_1}^{A_2B_2}\left(\phi_2\right)=-2$
from Corollary~\ref{cor:CLB}. This is in stark contrast with the SWAP operation, which cannot lead to creation of shared entanglement between Alice and Bob.
%Compared to the SWAP operation,
%the negativity of the entanglement cost is quite interesting,
%since Alice and Bob cannot share entanglement after the SWAP operation.

\begin{figure}
\centering
\includegraphics[width=.6\linewidth,trim=0cm 0cm 0cm 0cm]{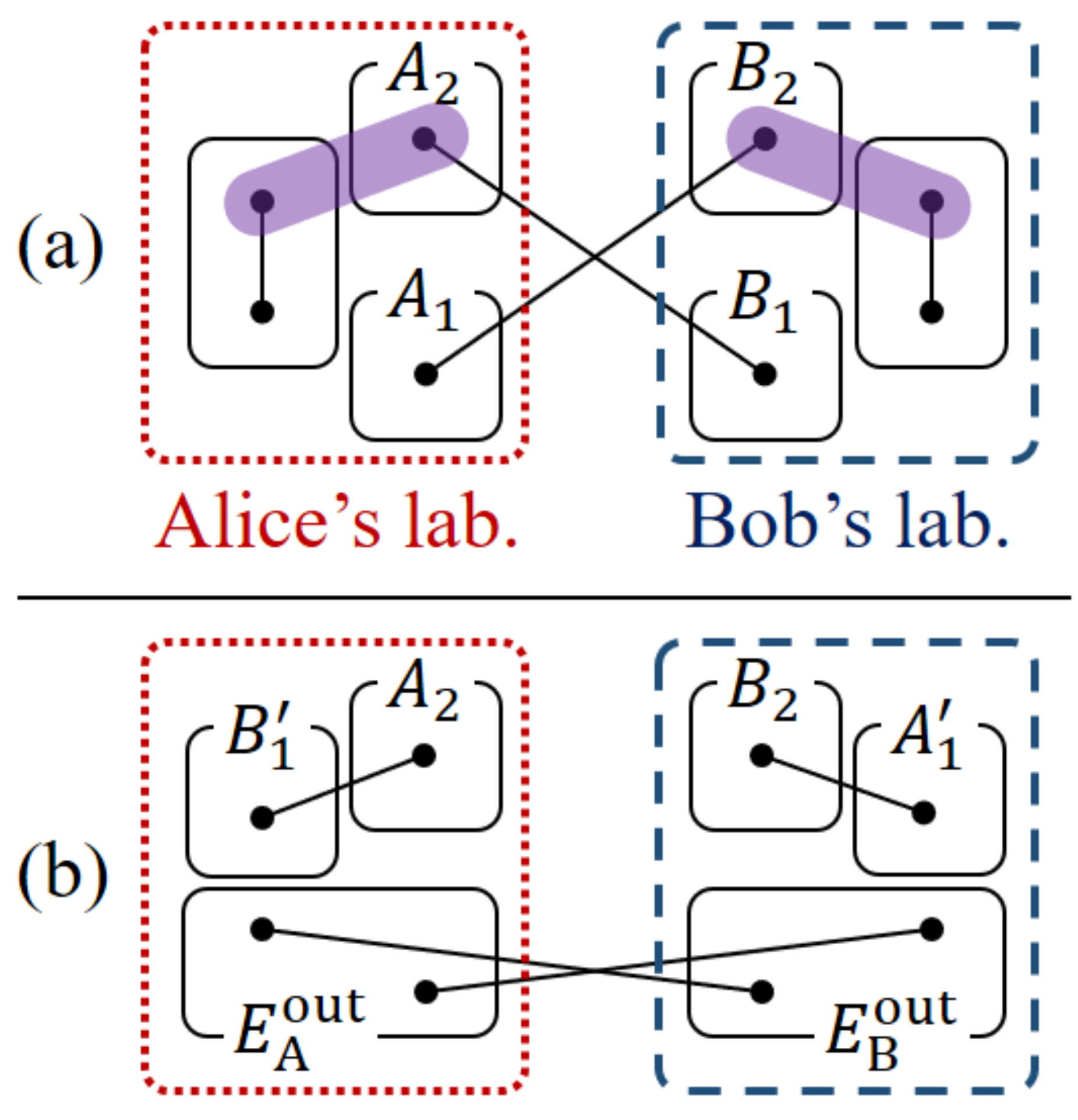}
\caption{
Illustration of the one-shot quantum state exchange protocol of $\ket{\phi_2}$ in Eq.~(\ref{eq:NOEC}).
(a) In order to exchange $A_1$ and $B_1$,
Alice and Bob locally prepare an ebit each, and they apply Bell measurements to the shaded areas.
(b) By performing local operations corresponding to the measurement outcomes,
the parts $A_1$ and $B_1$ can be exchanged.
At the same time,
Alice and Bob can share two ebits.
}
\label{fig:ES}
\end{figure}

From Proposition~\ref{prop:OEC2},
we can know that if $A_2$ and $B_2$ are symmetric,
then $\mathbf{e}_{A_1\leftrightarrow B_1}^{A_2B_2}\left(\psi\right)$ cannot be negative.
One may ask the question: Is there any condition that
implies the non-negativity of the optimal entanglement cost $\mathbf{e}_{A_1\leftrightarrow B_1}^{A_2B_2}$?
To answer this question,
we present the following inequalities.

\begin{Prop}  \label{prop:ineq}
\begin{eqnarray*}
\mathbf{e}_{A_1\leftrightarrow B_1}^{A_2B_2}(\psi)+
\mathbf{e}_{B'_1\leftrightarrow A'_1}^{A_2B_2}(\psi_{f_1})
&\ge&
0, \\
\mathbf{e}_{A_1\leftrightarrow B_1}^{A_2B_2}(\psi)+
\mathbf{e}_{A_2\leftrightarrow B_2}^{B'_1A'_1}(\psi_{f_1})
&\ge&
\mathbf{e}_{A\leftrightarrow B}(\psi),
\end{eqnarray*}
where $\mathbf{e}_{B'_1\leftrightarrow A'_1}^{A_2B_2}(\psi_{f_1})$
is the optimal entanglement cost for exchanging $B'_1$ and $A'_1$ when using $A_2$ and $B_2$,
and
$\mathbf{e}_{A_2\leftrightarrow B_2}^{B'_1A'_1}(\psi_{f_1})$ is
the optimal entanglement cost for exchanging $A_2$ and $B_2$
when using $B'_1$ and $A'_1$.
\end{Prop}

In Proposition~\ref{prop:ineq},
the first inequality comes from the fact that
Alice and Bob cannot increase the amount of entanglement between them
by means of LOCC~\cite{VPRK97},
while the second one is straightforward from the definitions
of the optimal entanglement costs.
From Proposition~\ref{prop:ineq},
we can see that
if $\mathbf{e}_{B'_1\leftrightarrow A'_1}^{A_2B_2}(\psi_{f_1})$
or $\mathbf{e}_{A_2\leftrightarrow B_2}^{B'_1A'_1}(\psi_{f_1})$ is non-positive,
then $\mathbf{e}_{A_1\leftrightarrow B_1}^{A_2B_2}(\psi)$ cannot be negative.
Moreover, if the condition $\mathbf{e}_{A_2\leftrightarrow B_2}^{B'_1A'_1}(\psi_{f_1})
\le
\mathbf{e}_{A\leftrightarrow B}(\psi)$ holds,
then Proposition~\ref{prop:ineq}
implies $\mathbf{e}_{A_1\leftrightarrow B_1}^{A_2B_2}(\psi)\ge0$.

In particular,
let us assume that $A_1$ and $B_1$ are symmetric.
Then
it is obvious that $0\ge \mathbf{e}_{A_1\leftrightarrow B_1}^{A_2B_2}(\psi)$,
from Proposition~\ref{prop:OEC1}.
If $0>\mathbf{e}_{A_1\leftrightarrow B_1}^{A_2B_2}(\psi)$
then it follows from Proposition~\ref{prop:ineq} that
$\mathbf{e}_{B'_1\leftrightarrow A'_1}^{A_2B_2}(\psi_{f_1}) >0$.
However,
since $B'_1$ and $A'_1$ are also symmetric,
Proposition~\ref{prop:OEC1} implies
$\mathbf{e}_{B'_1\leftrightarrow A'_1}^{A_2B_2}(\psi_{f_1}) \le 0$,
which leads to a contradiction.
Therefore,
we obtain the following corollary.

\begin{Cor} \label{cor:zero}
$\mathbf{e}_{A_1\leftrightarrow B_1}^{A_2B_2}\left(\psi\right)=0$,
if $A_1$ and $B_1$ %of $\ket{\psi}$ 
are symmetric.
\end{Cor}
This tells us that
if $A_1$ and $B_1$ are symmetric,
Alice and Bob cannot increase the amount of shared entanglement
after the OSQSE task,
even if they make use of the parts $A_2$ and $B_2$.

%%%%%%%%%%%%%%%%%%%%%%%%%%%%%%%%%%%%%%%%%%%%%%%%%%%%%%%%%%%%%%%%%%%%%%%%%%%%%%%%%%%%%%%%%%%%%%
%%%
%%%     Conclusion
%%%
%%%%%%%%%%%%%%%%%%%%%%%%%%%%%%%%%%%%%%%%%%%%%%%%%%%%%%%%%%%%%%%%%%%%%%%%%%%%%%%%%%%%%%%%%%%%%%
{\it Conclusion.}---
In this work,
we have considered a one-shot version of the original quantum state exchange task,
and have formally defined the OSQSE task and its optimal entanglement costs.
We have derived converse bounds on the optimal entanglement costs,
and have presented conditions on the initial state to achieve zero entanglement cost.
As a related open problem,
we can ask the following question:
If $\mathbf{e}_{A\leftrightarrow B}\left(\psi\right)=0$,
then is it possible to exchange the parts $A$ and $B$,
without classical communication and entanglement,
that is,
are there local operations $L_{\mathrm{A}}$ and $L_{\mathrm{B}}$
such that
$%\begin{equation}
\psi_{f_{12}}=\left(L_{\mathrm{A}}\otimes L_{\mathrm{B}}\right)(\psi)%? \nonumber
$%\end{equation}
?

The two counter-intuitive phenomena are the most interesting contribution of this work,
showing the major difference
between the SWAP operation and the OSQSE.
One phenomenon tells us that it is worth using the symmetric parts
in order to optimally perform the OSQSE.
The other shows that
the entanglement cost of the OSQSE can be negative.
By observing the aforementioned examples involving the phenomena,
we can provide another interesting open problem:
If $\mathbf{e}_{A_1\leftrightarrow B_1}^{A_2B_2}\left(\psi\right)\le0$,
do there exist Alice's and Bob's local operations $L'_{\mathrm{A}}$ and $L'_{\mathrm{B}}$
such that
$\psi_{f_1} \otimes {\Phi} = \left(L'_{\mathrm{A}}\otimes L'_{\mathrm{B}}\right)(\psi)$?

A further open problem is whether
the catalytic use of entanglement~\cite{JP99,EW00,MBDRC17}
can reduce the optimal entanglement cost for the OSQSE.
To be more specific,
for the initial state $\ket{\psi}$,
do there exist a bipartite entangled state $\ket{\psi_{\mathrm{c}}}_{A_3B_3}$
shared by Alice and Bob
and a OSQSE protocol
%\begin{equation}
$\mathcal{C}_{K,L}:
AA_3E_{\mathrm{A}}^{\mathrm{in}}\otimes
BB_3E_{\mathrm{B}}^{\mathrm{in}}
\longrightarrow B'A_3E_{\mathrm{A}}^{\mathrm{out}}\otimes
A'B_3E_{\mathrm{B}}^{\mathrm{out}}$ %\nonumber
%\end{equation}
such that
$\psi_{f_{12}} \otimes \psi_{\mathrm{c}} \otimes {\Phi}
=
\left(\mathcal{C}_{K,L}\otimes\mathds{1}_{R}\right)
\left( \psi \otimes \psi_{\mathrm{c}} \otimes \Psi \right)$
and
$\log K-\log L<\mathbf{e}_{A\leftrightarrow B}\left(\psi\right)$?

Theoretically,
the OSQSE is a powerful two-user quantum communication task,
which includes quantum teleportation~\cite{BBCJPW93}
and quantum state merging~\cite{HOW05,HOW06}
as special cases.
Practically,
this task can be a fundamental building block for applications involving multiple users,
such as distributed quantum computation~\cite{CEHM99,BDLMSS04}
and quantum network~\cite{CZKM97,AML16,PWD18}.
%,since tasks for multiple users can conceptually be decomposed into tasks for two users.

%%%%%%%%%%%%%%%%%%%%%%%%%%%%%%%%%%%%%%%%%%%%%%%%%%%%%%%%%%%%%%%%%%%%%%%%%%%%%%%%%%%%%%%%%%%%%%
%%%
%%%     Acknowledgment
%%%
%%%%%%%%%%%%%%%%%%%%%%%%%%%%%%%%%%%%%%%%%%%%%%%%%%%%%%%%%%%%%%%%%%%%%%%%%%%%%%%%%%%%%%%%%%%%%%
\begin{acknowledgments}
%{\it Acknowledgments.}---
We would like to thank Ryuji Takagi and Bartosz Regula for fruitful discussion.
This research was supported
by Basic Science Research Program through the National Research Foundation of
Korea (NRF) funded by the Ministry of Science and ICT (NRF-2019R1A2C1006337) and the MSIT (Ministry
of Science and ICT), Korea, under the ITRC (Information Technology Research Center) support program
(IITP-2019-2018-0-01402) supervised
by the IITP (Institute for Information \& communications Technology Promotion).
H. Y. acknowledges Grant-in-Aid for JSPS Research Fellow, JSPS KAKENHI Grant No. 18J10192,
Cross-ministerial Strategic Innovation Promotion Program (SIP) (Council for Science,
Technology and Innovation (CSTI)), and CREST (Japan Science and Technology Agency) JPMJCR1671.
G. A. acknowledges support from the ERC Starting Grant GQCOP (Grant Agreement No. 637352).
\end{acknowledgments}

\bibliography{OSQSE}

\clearpage
%\cleardoublepage{}
%\begin{widetext}
\newpage
\appendix
\setcounter{equation}{0}
\setcounter{page}{1}

%\section*{Supplemental Material}

%%%%%%%%%%%%%%%%%%%%%%%%%%%%%%%%%%%%%%%%%%%%%%%%%%%%%%%%%%%%%%%%%%%%%%%%%%%%%%%%%%%%%%%%%%%%%%%%%%%
%%%
%%%   Appendix
%%%
%%%%%%%%%%%%%%%%%%%%%%%%%%%%%%%%%%%%%%%%%%%%%%%%%%%%%%%%%%%%%%%%%%%%%%%%%%%%%%%%%%%%%%%%%%%%%%%%%%%

%%%%%%%%%%%%%%%%%%%%%%%%%%%%%%%%%%%%%%%%%%%%%%%%%%%%%%%%%%%%%%%%%%%%%%%%%%%%%%%%%%%%%%%%%%%%%%%%%%%
%%%   Proof of Corollary~\ref{cor:Com_lb}
%%%%%%%%%%%%%%%%%%%%%%%%%%%%%%%%%%%%%%%%%%%%%%%%%%%%%%%%%%%%%%%%%%%%%%%%%%%%%%%%%%%%%%%%%%%%%%%%%%%
\section{Proof of Corollary~\ref{cor:CLB}} \label{app:pf_CLB}

We show that there exists a number $\alpha_0\in[0,\infty]$
such that $l_{1|2}^c(\psi)=f_{\psi}(\alpha_0)$.
Note that
the function $f_{\psi}(\alpha)$ is continuous on the compact set $[0,1]$.
So the extreme value theorem implies
that there exists a number $\alpha_1\in[0,1]$ such that $f_{\psi}(\alpha_1)\ge f_{\psi}(\alpha)$
for all $\alpha\in[0,1]$.
Let us consider the function $g(x)$ on the interval $[0,1]$ defined as
\begin{equation}
g(x)=
\begin{cases}
f_{\psi}(\infty)
&\text{if $x=0$} \\
f_{\psi}(\frac{1}{x})
&\text{otherwise},
\end{cases} \nonumber
\end{equation}
then $g(x)$ is continuous on $[0,1]$.
By using the extreme value theorem again,
there exists a number $x_0\in[0,1]$ such that $g(x_0)\ge g(x)$
for all $x\in[0,1]$.
It follows that
there exists a number $\alpha_2\in[1,\infty]$ such that $f_{\psi}(\alpha_2)\ge f_{\psi}(\alpha)$
for all $\alpha\in[1,\infty]$.
By setting $\alpha_0=\max\left\{\alpha_1,\alpha_2\right\}$,
we obtain that
$l_{1|2}^c(\psi)=f_{\psi}(\alpha_0)\ge f_{\psi}(\alpha)$
for all $\alpha\in[0,\infty]$.

%%%%%%%%%%%%%%%%%%%%%%%%%%%%%%%%%%%%%%%%%%%%%%%%%%%%%%%%%%%%%%%%%%%%%%%%%%%%%%%%%%%%%%%%%%%%%%%%%%%
%%%   Proof of Theorem~\ref{thm:SC0}
%%%%%%%%%%%%%%%%%%%%%%%%%%%%%%%%%%%%%%%%%%%%%%%%%%%%%%%%%%%%%%%%%%%%%%%%%%%%%%%%%%%%%%%%%%%%%%%%%%%
\section{Proof of Theorem~\ref{thm:SC0}} \label{app:pf_thm:SC0}

When $X=A$ and $Y=B$,
consider the Schmidt decompositions of $\ket{\psi}$,
\begin{equation}
\ket{\psi}_{ABR}
=
\sum_{i=1}^{N}\sqrt{\lambda_i}\ket{\xi_i}_{AB}\otimes\ket{\iota_i}_{R}, \nonumber
\end{equation}
where $\lambda_i>0$ with $\sum_{i=1}^{N}\lambda_i=1$.
For the computational bases $\{\ket{j}\}$ and $\{\ket{k}\}$ on the systems $A$ and $B$,
respectively,
we have
\begin{eqnarray}
&&\ket{\psi}_{ABR} \nonumber \\
&&=
\sum_{i=1}^{N}\sqrt{\lambda_i}\sum_{j,k}[\Omega_{AB}^i(\psi)]_{jk}\ket{j}_{A}\otimes\ket{k}_{B}\otimes\ket{\iota_i}_{R}, \nonumber
\end{eqnarray}
where $[\Omega_{AB}^i(\psi)]_{jk}
=\left(\bra{j}_{A}\otimes\bra{k}_{B}\right)\ket{\xi_i}_{AB}$.
If the parts $A$ and $B$ are perfectly exchanged,
then Alice and Bob hold the final state
\begin{eqnarray}
&&\ket{\psi}_{BAR} \nonumber \\
&&=
\sum_{i=1}^{N}\sqrt{\lambda_i}\sum_{j,k}[\Omega_{AB}^i(\psi)]_{kj}\ket{j}_{B}\otimes\ket{k}_{A}\otimes\ket{\iota_i}_{R}, \nonumber
\end{eqnarray}
By the hypothesis,
there exist isometries $U$ and $V$ such that for each $i$,
\begin{equation}
{\left(\Omega_{AB}^i(\psi)\right)}^t=U \Omega_{AB}^i(\psi)V. \nonumber
\end{equation}
So we have, for each $i$,
\begin{equation}
[\Omega_{AB}^i(\psi)]_{kj}
=\sum_{l,m}[\Omega_{AB}^i(\psi)]_{lm}
\bra{j}U\ket{l}\bra{k}V^t\ket{m}, \nonumber
\end{equation}
which implies that
\begin{eqnarray}
&&\ket{\psi}_{BAR} \nonumber \\
&&=
\sum_{i=1}^{N}\sqrt{\lambda_i}\sum_{l,m}[\Omega_{AB}^i(\psi)]_{lm}
\sum_{j}\ket{j}\bra{j}U\ket{l} \nonumber \\
&&\quad\otimes\sum_{k}\ket{k}\bra{k}V^t\ket{m}\otimes\ket{\iota_i}_{R}, \nonumber \\
&&=
\sum_{i=1}^{N}\sqrt{\lambda_i}\sum_{l,m}[\Omega_{AB}^i(\psi)]_{lm}
U\ket{l}\otimes V^t\ket{m}\otimes\ket{\iota_i}_{R}, \nonumber \\
&&=\left(U\otimes V^t\otimes I_R\right)\ket{\psi}_{ABR}. \nonumber
\end{eqnarray}
Hence,
$\mathbf{e}_{A\leftrightarrow B}\left(\psi\right)=0$.

Similarly,
we can show that $\mathbf{e}_{A_1\leftrightarrow B_1}\left(\psi\right)=0$
by using isometries $U'$ and $V'$ such that for each $i$,
${\left(\Omega_{A_1B_1}^i(\psi)\right)}^t=U' \Omega_{A_1B_1}^i(\psi)V'$.

%%%%%%%%%%%%%%%%%%%%%%%%%%%%%%%%%%%%%%%%%%%%%%%%%%%%%%%%%%%%%%%%%%%%%%%%%%%%%%%%%%%%%%%%%%%%%%%%%%%
%%%   Proof of Theorem~\ref{thm:SCP}
%%%%%%%%%%%%%%%%%%%%%%%%%%%%%%%%%%%%%%%%%%%%%%%%%%%%%%%%%%%%%%%%%%%%%%%%%%%%%%%%%%%%%%%%%%%%%%%%%%%
\section{Proof of Theorem~\ref{thm:SCP}} \label{app:pf_thm:SCP}

We use the following lemma in order to prove Theorem~\ref{thm:SCP}.

\begin{Lem} \label{lem:REtoPD}
Let $Z$ and $W$ be any discrete random variables
on alphabets $\mathcal{Z}$ and $\mathcal{W}$
with $|\mathcal{Z}|=N$ and $|\mathcal{W}|=M$.
Let $\{p_i\}_{i=1}^{N}$ and $\{q_i\}_{i=1}^{M}$ be probability distributions for $X$ and $Y$,
respectively.
If the following equality holds for all $\alpha\in[0,\infty]$,
\begin{equation*}
H_\alpha(Z)= H_\alpha(W),
\end{equation*}
where $H_\alpha(\cdot)$ is the R\'enyi entropy of classical random variables,
then $|\mathcal{Z}|=|\mathcal{W}|$
and there exists a permutation $\sigma\in S_{N}$
such that $p_i=q_{\sigma(i)}$ for all $i\in[N]$,
where $S_N$ is the set of all permutations on $[N]=\{1,\cdots,N\}$.
\end{Lem}

Note that $H_{\alpha}(Z)=\lim_{x\to\alpha}H_{x}(Z)$
and $S_{\alpha}(\rho_A)=\lim_{x\to\alpha}S_{x}(\rho_A)$
for each $\alpha\in[0,\infty]$.

\begin{proof}
Suppose that
$H_\alpha(Z)= H_\alpha(W)$ for all $\alpha\in[0,\infty]$.
Since $H_0(Z)= H_0(W)$,
it holds that
$\left|\mathcal{Z}\right|=\left|\mathcal{W}\right|$.
For convenience,
we assume that any probability distribution $\{r_i\}_{i=1}^{N}$
satisfies $r_1\ge r_i$ for all $i\in[N]$.

We now prove the statement by using mathematical induction on $N$.

(i)
If $N=2$,
then $H_\infty(Z)= H_\infty(W)$ implies $p_1=q_1$
and so $p_2=1-p_1=1-q_1=q_2$.
Thus the statement is true.

(ii)
Suppose that the statement is true for $N=k-1$.
Let $Z_1$ and $Z_2$ be discrete random variables
on alphabets $\mathcal{Z}_1$ and $\mathcal{Z}_2$
with $|\mathcal{Z}_1|=|\mathcal{Z}_2|=k$.
Let $\{p_i\}_{i=1}^{k}$ and $\{q_i\}_{i=1}^{k}$
be probability distributions for $Z_1$ and $Z_2$,
respectively.
Since $H_\infty(Z_1)= H_\infty(Z_2)$, $p_1=q_1$.
By setting $p'_i=\frac{p_{i+1}}{1-p_1}$
and $q'_i=\frac{q_{i+1}}{1-p_1}$
for each $i\in[k-1]$,
we can construct random variables $Z'_1$ and $Z'_2$
on alphabets $\mathcal{Z}'_1$ and $\mathcal{Z}'_2$
whose probability distributions are
$\{p'_i\}_{i=1}^{k-1}$ and $\{q'_i\}_{i=1}^{k-1}$,
respectively.
Obviously, $\left|\mathcal{Z}'_1\right|=\left|\mathcal{Z}'_2\right|=k-1$,
and so $H_0(Z'_1)= H_0(Z'_2)$.
Observe that
for $\alpha\in(0,1)\cup(1,\infty)$
\begin{eqnarray}
&&H_\alpha(Z_1)= H_\alpha(Z_2) \nonumber \\
&\implies&
\frac{1}{1-\alpha}\log \left(\sum_{i=1}^{k}p_i^{\alpha}\right)
=\frac{1}{1-\alpha}\log \left(\sum_{i=1}^{k}q_i^{\alpha}\right) \nonumber \\
&\implies&
\sum_{i=2}^{k}p_i^{\alpha}
=\sum_{i=2}^{k}q_i^{\alpha} \nonumber \\
&\implies&
\sum_{i=1}^{k-1}\left(\frac{p_{i+1}}{1-p_1}\right)^{\alpha}
=\sum_{i=1}^{k-1}\left(\frac{q_{i+1}}{1-p_1}\right)^{\alpha} \nonumber \\
&\implies&
\frac{1}{1-\alpha}\log \left(\sum_{i=1}^{k-1}(p'_i)^{\alpha}\right)
=\frac{1}{1-\alpha}\log \left(\sum_{i=1}^{k-1}(q'_i)^{\alpha}\right) \nonumber \\
&\implies&
H_\alpha(Z'_1)= H_\alpha(Z'_2). \nonumber
\end{eqnarray}
In addition,
if $\alpha=1$,
then
\begin{eqnarray}
&&H_1(Z_1)= H_1(Z_2) \nonumber \\
&\implies&
\sum_{i=1}^{k}p_i\log \frac{1}{p_i}
=\sum_{i=1}^{k}q_i\log \frac{1}{q_i} \nonumber \\
&\implies&
\sum_{i=2}^{k}p_i\log \frac{1}{p_i}
=\sum_{i=2}^{k}q_i\log \frac{1}{q_i} \nonumber \\
&\implies&
\left(1-p_1\right)\log\left(1-p_1\right)
+\sum_{i=2}^{k}p_i\log \frac{1}{p_i}  \nonumber \\
&&=\left(1-p_1\right)\log\left(1-p_1\right)
+\sum_{i=2}^{k}q_i\log \frac{1}{q_i} \nonumber \\
&\implies&
\sum_{i=2}^{k}\frac{p_i}{1-p_1}\log \frac{1-p_1}{p_i}
=\sum_{i=2}^{k}\frac{q_i}{1-p_1}\log \frac{1-p_1}{\tau_i} \nonumber \\
&\implies&
\sum_{i=1}^{k-1}p'_i\log \frac{1}{p'_i}
=\sum_{i=1}^{k-1}q'_i\log \frac{1}{q'_i} \nonumber \\
&\implies&
H_1(Z'_1)= H_1(Z'_2). \nonumber
\end{eqnarray}
Finally,
we have
\begin{eqnarray}
H_\infty(Z'_1)- H_\infty(Z'_2)
&=&\lim_{\alpha\to\infty}H_\alpha(Z'_1)-\lim_{\alpha\to\infty}H_\alpha(Z'_2) \nonumber \\
&=&\lim_{\alpha\to\infty}\left(H_\alpha(Z'_1)-H_\alpha(Z'_2) \right)=0. \nonumber
\end{eqnarray}
It follows that $H_\alpha(Z'_1)= H_\alpha(Z'_2)$ for all $\alpha\in[0,\infty]$.
By the induction hypothesis,
there exists a permutation $\sigma'\in S_{k-1}$
such that $p'_i=q'_{\sigma'(i)}$ for all $i\in[k-1]$.
Define $\sigma(1)=1$ and $\sigma(i)=\sigma'(i-1)$ with $i\neq1$.
Then $\sigma\in S_k$ and $p_i=q_{\sigma(i)}$ for all $i\in[k]$.
Therefore,
the statement is true for $N=k$.
\end{proof}

In fact,
we can prove Lemma~\ref{lem:REtoPD}
by assuming a weaker condition as follows.
Let $S$ be a subset of $[0,\infty]$
including 0,
the extended real number $\infty$,
and a sequence $\{s_n\}_{n\in\mathbb{N}}$
such that $\lim_{n\to\infty}s_n=\infty$.
Then we can show that
if $H_\alpha(Z)= H_\alpha(W)$ holds for all $\alpha\in S$,
then $Z$ and $W$ have the same probability distribution.

The contrapositive of the following lemma proves Theorem~\ref{thm:SCP}.

\begin{Lem}[Sufficient conditions on the initial state $\ket{\psi}$ with $\mathbf{e}_{X\leftrightarrow Y}\left(\psi\right)>0$] \label{lem:SCp}
Let $(X,Y)$ be the pair of two systems,
which can be either $(A_1,B_1)$ or $(A,B)$.
Let $\{\lambda_i\}_{i=1}^N$ and $\{\tau_i\}_{i=1}^M$ be non-zero eigenvalues
for the reduced states $\rho_X$ and $\rho_Y$ of $\ket{\psi}$,
respectively,
which satisfy $\lambda_{1}\ge\cdots\ge\lambda_{N}$, $\tau_{1}\ge\cdots\ge\tau_{M}$,
and $\sum_{i=1}^{N}\lambda_i=\sum_{i=1}^{M}\tau_i=1$.
Then $\mathbf{e}_{X\leftrightarrow Y}>0$,
if one of the following conditions holds:

(i) $N\neq M$.

(ii) $N=M$ and $\lambda_{i'}\neq\tau_{i'}$ for some $i'\in[N]=\{1,\cdots,N\}$.
\end{Lem}

\begin{proof}
(i)
If $N\neq M$,
then $\mathrm{rank}(\rho_{X})\neq\mathrm{rank}(\rho_{Y})$,
which means
\begin{equation*}
\mathbf{e}_{X\leftrightarrow Y}\left(\psi\right)\ge\left|S_0(\rho_X)-S_0(\rho_Y)\right|>0,
\end{equation*}
by the converse bound in Eq.~(\ref{eq:CLB}).

(ii)
Suppose that $\ket{\psi}$ satisfies $N=M$ and $\lambda_{i'}\neq\tau_{i'}$ for some $i'\in[N]$.
Let $Z$ and $W$ be discrete random variables on alphabets $\mathcal{Z}$ and $\mathcal{W}$
with $|\mathcal{Z}|=|\mathcal{W}|=N$,
whose probability distributions are $\{\lambda_i\}_{i=1}^N$ and $\{\tau_i\}_{i=1}^N$,
respectively.
Let us consider the set
\begin{equation*}
A=\{i\in[N]|\lambda_{i}\neq\tau_{i}\},
\end{equation*}
then $A$ is a non-empty subset of $[N]$,
since $i'\in A$.
So we can choose the largest element in $A$,
say $j$.
Then $\lambda_{j}\neq\tau_{j}$ and $\lambda_{i}=\tau_{i}$ for all $i>j$
by the definition of the set $A$.
If $\lambda_{j}>\tau_{j}$ (or $\lambda_{j}<\tau_{j}$)
than $\lambda_{i}>\tau_{j}$ (or $\lambda_{j}<\tau_{i}$) for all $i\in[j]$.
Thus $\lambda_{i}\neq\tau_{j}$ (or $\lambda_{j}\neq\tau_{i}$) for all $i\in[j]$,
which shows that
for each $\sigma\in S_j$,
there exists $i\in[j]$ such that $\lambda_{i}\neq\tau_{\sigma(i)}$.
It follows that
for each $\sigma\in S_N$,
there exists $i\in[N]$ such that $\lambda_{i}\neq\tau_{\sigma(i)}$.
From the contrapositive of Lemma~\ref{lem:REtoPD},
there exists $\alpha'\in[0,\infty]$ such that
$H_{\alpha'}(X)\neq H_{\alpha'}(Y)$.
Therefore,
from the converse bound in Eq.~(\ref{eq:CLB}),
we obtain
\begin{equation*}
\mathbf{e}_{X\leftrightarrow Y}\left(\psi\right)\ge\left|S_{\alpha'}(\rho_X)-S_{\alpha'}(\rho_Y)\right|>0.
\end{equation*}
\end{proof}

%\end{widetext}
\end{document}